\documentclass[a4paper,UKenglish,cleveref, autoref, thm-restate]{lipics-v2021}
\hideLIPIcs  %uncomment to remove references to LIPIcs series (logo, DOI, ...), e.g. when preparing a pre-final version to be uploaded to arXiv or another public repository

\title{Approximating $\delta$-Dispersion} %TODO Please add

\author{Tom Janßen}{Department of Computer Science, RWTH Aachen University, Germany}{janssen@algo.rwth-aachen.de}{https://orcid.org/0000-0003-4617-3540}{Funded by the German Research Foundation (DFG) – WO 1451/2-1}
\authorrunning{T. Janßen} %TODO mandatory. First: Use abbreviated first/middle names. Second (only in severe cases): Use first author plus 'et al.'

\Copyright{Tom Janßen} %TODO mandatory, please use full first names. LIPIcs license is "CC-BY";  http://creativecommons.org/licenses/by/3.0/

\ccsdesc[500]{Theory of computation~Approximation algorithms analysis}
\ccsdesc[500]{Theory of computation~Problems, reductions and completeness}

\keywords{Facility Location, Approximation Algorithms, Approximation Hardness, Independent Set} %TODO mandatory; please add comma-separated list of keywords

\category{} %optional, e.g. invited paper

\relatedversion{} %optional, e.g. full version hosted on arXiv, HAL, or other respository/website
\nolinenumbers %uncomment to disable line numbering

\EventEditors{John Q. Open and Joan R. Access}
\EventNoEds{2}
\EventLongTitle{42nd Conference on Very Important Topics (CVIT 2016)}
\EventShortTitle{CVIT 2016}
\EventAcronym{CVIT}
\EventYear{2016}
\EventDate{December 24--27, 2016}
\EventLocation{Little Whinging, United Kingdom}
\EventLogo{}
\SeriesVolume{42}
\ArticleNo{23}
\usepackage{mathtools}
\usepackage{xparse}

\usepackage{tikz}
\usetikzlibrary{positioning,calc,arrows.meta,decorations.pathreplacing,angles,quotes,shapes.geometric,shapes.misc}

\newcommand{\oball}[2]{\ensuremath{B^<\left(#1, #2\right)}}
\newcommand{\cball}[2]{\ensuremath{B^\leq\left(#1, #2\right)}}

\newcommand{\half}{\ensuremath{\tfrac{1}{2}}}

\newcommand{\OPT}{\operatorname{OPT}}
\newcommand{\wreath}{\operatorname{wr}}

\newcommand{\PTIME}{{\sf P}}
\newcommand{\NP}{{\sf NP}}

\newcommand{\XP}{{\sf XP}}
\newcommand{\APX}{{\sf APX}}
\newcommand{\pAPX}{{\sf poly-APX}}

\NewDocumentCommand{\disp}{O{\delta}}{\ensuremath{#1}\text{-}\textsc{Dispersion}}
\NewDocumentCommand{\autodisp}{O{\delta}}{\ensuremath{#1}\text{-}\textsc{AutoDispersion}}
\NewDocumentCommand{\dispn}{O{\delta}O{G}}{\ensuremath{#1}\textup{-disp}\ensuremath{(#2)}}
\NewDocumentCommand{\autodispn}{O{\delta}O{G}}{\ensuremath{#1}\textup{-auto-disp}\ensuremath{(#2)}}

\begin{document}

\maketitle

\begin{abstract}
 We consider a continuous facility location problem called \disp.
For some fixed $\delta > 0$, the goal is to place as many facilities on a graph as possible with pairwise distance at least $\delta$.
The facilities may be located on the vertices of the graph, or the interior of the edges.
This problem can be interpreted as a continuous version of the well-known {\sc Independent Set} problem.
Its approximation behavior is very similar for large values of $\delta$.
Notably, Grigoriev et al.~[Algorithmica 21] showed that \disp{} is solvable in polynomial time when $\delta = \frac{1}{x}$ or $\delta = \frac{2}{x}$ for a natural number $x$ and \NP-hard otherwise.

We study the approximability of \disp{} depending on the value of $\delta$.
For $\delta > 2$, we show \pAPX-hardness, while for all $\delta < 2$ that are not solvable in polynomial time we show \APX-hardness.
Thanks to a translation theorem for $\delta$ due to Hartmann et al.~[MFCS 22], we may focus our attention for approximation algorithms on the intervals $(\frac{2}{3}, 1)$ and $(1, 2)$.
We provide several approximation algorithms with an approximation factor approaching $1$ as $\delta$ approaches one of the interval boundaries.
Surprisingly, the behavior as $\delta$ approaches $\frac{2}{3}$ from above is different:
As our hardness reductions reveal, it is impossible (under standard complexity-theoretic assumptions) to construct an approximation algorithm with an approximation factor approaching $1$ as $\delta$ approaches $\frac{2}{3}$ from above.
\end{abstract}
\newpage

\section{Introduction}

This work investigates the approximability of a continuous facility location problem. 
In this model, the input is a graph with unit-length edges, where a \emph{point} can be located either on a vertex or at any real distance \(0 \leq \lambda \leq 1\) along an edge. 
The continuum set of points on all vertices and edges is denoted with $P(G)$, and the distance of two points is the length of the shortest path in the underlying metric space.
Originally introduced by Dearing \& Francis~\cite{Dearing1974}, this model has been extensively studied under various optimization objectives.

\disp{} has a rational parameter $\delta > 0$.
Then, a subset $S \subseteq P(G)$ is \emph{$\delta$-dispersed}, if for any two points $p,q \in S$ it holds that their distance $d(p,q) \geq \delta$.
Given a graph \(G\) and a fixed \(\delta > 0\), the objective is to find a \(\delta \)-dispersed set of maximum cardinality. 
The size of a maximum $\delta$-dispersed set $S$ is denoted with $\dispn = |S|$.
The decision version of this problem is known as \disp.

Both the standard complexity and parameterized complexity of \disp{} are well understood.
\disp{} is polynomial-time solvable when $\delta$ can be written as $\frac{1}{x}$ or $\frac{2}{x}$ for a natural number $x > 0$, and \NP-hard otherwise, as shown by Grigoriev et al.~\cite{DBLP:journals/algorithmica/GrigorievHLW21}. 
A later work by Hartmann et al.~\cite{DBLP:conf/mfcs/HartmannL22} settled the parameterized complexity with solution size as parameter, as well as other natural graph parameters.
For $\delta < 2$, \disp{} is fixed-parameter tractable with solution size as parameter, while for $\delta > 2$ it becomes $W[1]$-hard, due to its similarities to Independent Set.

Given these boundaries, exploring the approximability of \(\delta \)-Dispersion is a natural next step. 
Tamir~\cite{Tamir1991} previously examined a closely related dual problem where the number of points is fixed and \(\delta \) is maximized. 
However, the approximability of \disp{} across different (but fixed) values of \(\delta \) remains unexplored.
We provide a comprehensive analysis of these approximation properties for all rational \(\delta > 0\). 
Because of the problem's inherent connection to Independent Set, positive approximation results are highly unlikely when \(\delta > 2\). 
Conversely, because any \(\delta \)-dispersed set remains valid for any \(\delta' < \delta\), the polynomial-time solvable cases provide a natural approximation mechanism for \(\delta < 2\). 
Consequently, as \(\delta \) approaches a polynomial-time solvable value of $\delta$ from below, we expect the achievable approximation ratio converge to 1.

\paragraph*{Further Related Work}

Standard literature on facility location includes the books by Drezner~\cite{drezner1996facility} and Mirchandani \& Francis~\cite{mirchandani1990discrete}. 
Similarly to how {\sc Independent Set} and {\sc Vertex Cover} are related, a dual problem to \disp{} is the {\sc $\delta$-Covering} problem, which aims to cover all points of a graph using the minimum number of points.
For this problem, Tamir~\cite{Tamir87} and Megiddo \& Tamir~\cite{MegiddoT1983} provided polynomial-time algorithms for special graph classes and proved \NP-hardness for $\delta = 2$. 
Subsequently, Hartmann et al.~\cite{DBLP:journals/mp/HartmannLW22} demonstrated that {\sc $\delta$-Covering} is solvable in polynomial time for unit fraction values of $\delta$ and remains \NP-hard otherwise, while also resolving the parameterized complexity of the problem for all rational $\delta > 0$ with the natural parameter solution size. 
A work by Hartmann \& Janßen~\cite{DBLP:conf/waoa/HartmannJ24} analyzed the approximation properties of {\sc $\delta$-Covering} for all rational $\delta > 0$. 
In the same model, Frei et al.~\cite{DBLP:conf/isaac/FreiGHHM24} investigated a different problem known as {\sc $\delta$-Tour}, where the objective is to traverse the graph such that every point lies within a distance at most $\delta$ from the tour. 
They established its \NP-hardness for all rational $\delta > 0$ and analyzed its approximation properties and parameterized complexity. 
In \cite{DBLP:conf/stacs/HartmannM25}, Hartmann et al.\ study independent set and dominating set problems in bounded treewidth graphs with integer and real distances.

\paragraph*{Our Results}

\newcommand{\putmarkx}[1]{
	\pgfmathsetmacro{\xval}{(\x-0.6667)*7.5}
	\draw (\xval,0) -- (\xval,-0.15);
	\node (putmarkx) at (\xval, -0.5) {$#1$};
}
\newcommand{\putmarky}[1]{
	\pgfmathsetmacro{\yval}{(\y-1)*2.5}
	\draw (-0.2, \yval) -- (0, \yval);
	\node (putmarky) at (-0.5+0.1*\c, \yval) {$#1$};
}

\begin{figure}[t]
	\centering
	\begin{tikzpicture}[yscale=1.0,xscale=1.0]
		\def\colora{blue}
		\def\colorb{teal}
		\def\colorc{purple}
		\def\colord{orange}
		% coord system
		\draw[->, ultra thick] (-0.2, 0) -- (10.5, 0);
		\node at (11, -0.5) {$\delta$};
		\draw[->, ultra thick] (0, -0.2) -- (0, 3);
		\node[rotate=90] at (-1, 2) {ratio};
		
		\def\x{1}\putmarkx{1}
        \def\x{2}\putmarkx{2}
		\foreach \a/\b in {2/3, 4/5, 6/7, 3/2, 5/4, 7/6, 5/3, 9/5, 13/7}{
			\def\x{\a/\b}\putmarkx{\frac{\a}{\b}}
		}
		
		\def\c{0}
        \def\y{2}\putmarky{2}
        \def\y{1}\putmarky{1}
		\foreach \a/\b/\c in {4/3/1, 3/2/0} {
			\def\y{\a/\b}\putmarky{\frac{\a}{\b}}
		}

        % 2 approx baseline
		\pgfmathsetmacro{\from}{(1.5-0.6667)*7.5}
		\pgfmathsetmacro{\to}{(2-0.6667)*7.5}
		\draw[-, color=\colord, ultra thick] (\from, 2.5) -- node[above, xshift=30] {\Cref{lemma:approxLargerThreeHalves}} (\to, 2.5);
        
		% constructions
		\foreach \x in {1, ..., 100} {
			% approaching 2
			\pgfmathsetmacro{\from}{((((4*\x+1)/(2*\x+1))-0.6667)*7.5}
			\pgfmathsetmacro{\to}{(((((4*\x+5)/(2*\x+3))-0.6667)*7.5}
			\pgfmathsetmacro{\val}{(((\x+1)/\x)-1)*2.5}
			\ifthenelse{\x=1}{
				\draw[-, color=\colora, thick] (\from, \val) -- node[below, yshift=-1] {\Cref{lemma:approxLargerOneApproachingTwo}} (\to, \val);
			}{
				\draw[-, color=\colora, ultra thick] (\from, \val) -- (\to, \val);
			}
            \ifthenelse{\x<4}{
                \draw[-, color=\colora, thick] (\to, \val-0.1) -- (\to, \val+0.1);
            }{}
			
			% approaching 1 from above
			\pgfmathsetmacro{\from}{((((2*\x+1)/(2*\x))-0.6667)*7.5}
			\pgfmathsetmacro{\to}{(((((2*\x+3)/(2*\x+2))-0.6667)*7.5}
			\pgfmathsetmacro{\val}{(((\x+1)/\x)-1)*2.5}
			\ifthenelse{\x=1}{
				\draw[-, color=\colorb, ultra thick] (\from, \val) -- node[above, yshift=-1] {\Cref{lemma:approxLargerOneApproachingOne}} (\to, \val);
			}{
				\draw[-, color=\colorb, ultra thick] (\from, \val) -- (\to, \val);
			}
            \ifthenelse{\x<5}{
                \draw[-, color=\colorb, thick] (\from, \val-0.1) -- (\from, \val+0.1);
            }{}
			
			% approaching 1 from below
			\pgfmathsetmacro{\from}{((((2*\x)/(2*\x+1))-0.6667)*7.5}
			\pgfmathsetmacro{\to}{(((((2*\x+2)/(2*\x+3))-0.6667)*7.5}
			\pgfmathsetmacro{\val}{(((\x+2)/(\x+1))-1)*2.5}
			\ifthenelse{\x=1}{
				\draw[-, color=\colorc, ultra thick] (\from, \val) -- node[above, yshift=-1, xshift=10] {\Cref{lemma:approxSmallerOne}} (\to, \val);
			}{
				\draw[-, color=\colorc, ultra thick] (\from, \val) -- (\to, \val);
			}
            \ifthenelse{\x<4}{
                \draw[-, color=\colorc, thick] (\to, \val-0.1) -- (\to, \val+0.1);
            }{}
 		}

	\end{tikzpicture}
	\caption{Upper bounds on the approximation ratio of \disp{} plotted for $\delta \in (\frac{2}{3},2)$. Small vertical lines indicate that that interval end is closed, while the others are open. For better visibility, these small vertical lines have been omitted for the very short intervals.}
	\label{fig:approximations}
\end{figure}
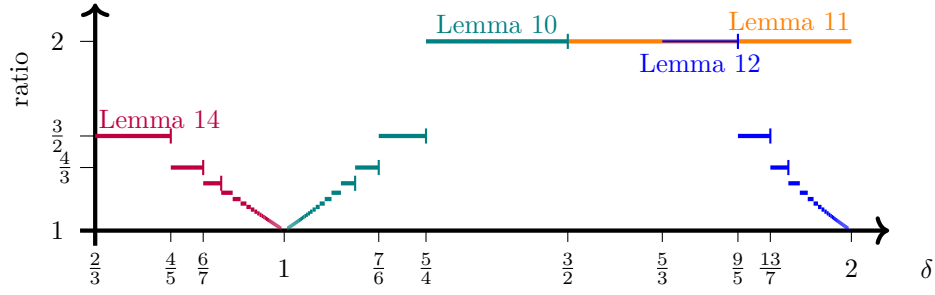

We fully analyze the approximation hardness of \disp.
For $\delta > 2$, we show \pAPX-hardness, while for all $\delta \leq 2$ that cannot be written as $\frac{1}{x}$ or $\frac{2}{x}$ for some natural number $x$, we show \APX-hardness.
We further show that, for all $\delta \leq 2$, the problem admits constant-factor approximations.
As mentioned before, due to the fact that \disp{} is solvable in polynomial time for all $\delta$ with $\delta = \frac{1}{x}$ or $\delta = \frac{2}{x}$ for some natural number $x$, we expect that the approximation factor approaches $1$ as $\delta$ approaches a polynomial-time solvable case.
\begin{itemize}
    \item For the interval $(1, 2)$, this is true both when $\delta$ approaches $1$ from above or $2$ from below. 
    In particular, we provide an approximation algorithm for $\delta < \frac{2k+1}{2k}$ that achieves an approximation factor of $\frac{k+1}{k}$, and show that for $\delta > \frac{4k+1}{2k+1}$ a $\delta$-dispersed set is at most $\frac{k+1}{k}$ times larger than a $2$-dispersed set, which can be computed in polynomial time.
    \item Maybe surprisingly, for the interval $(\frac{2}{3}, 1)$ this is only true when $\delta$ approaches $1$ from below.
    We show that for $\delta > \frac{2k}{2k+1}$ a $\delta$-dispersed set is at most $\frac{k+2}{k+1}$ times larger than a $1$-dispersed set, which can be computed in polynomial time.
    On the other hand, our \APX-hardness reductions yield a lower bound of $(1+c)$ on the approximation ratio for the entire interval $(\frac{2}{3}, \frac{3}{4}]$ for some fixed $c \in \mathbb R^+$, ruling out the existence of efficient algorithms with an approximation factor approaching $1$ as $\delta$ approaches $\frac{2}{3}$ from above.
\end{itemize}
Except for \Cref{lemma:approxLargerOneApproachingTwo}, we also show that the approximation ratios we achieve are best possible using our approaches.
The approximation ratios we achieve are plotted in \Cref{fig:approximations}.
Due to a translation result for $\delta$ by Hartmann et al.~\cite{DBLP:conf/mfcs/HartmannL22}, we may focus on the intervals $(\frac{2}{3}, 1)$ and $(1, 2)$ and use their result to extend our bounds to all other smaller intervals in $(0, \frac{2}{3})$.

The paper is organized as follows.
In \Cref{sec:preliminaries}, we introduce basic notations and existing problem definitions.
In \Cref{sec:approximations}, we present our approximation schemes, and in \Cref{sec:hardness}, we present our approximation hardness results. 
\section{Preliminaries}
\label{sec:preliminaries}

We use $\mathbb N^+$ and $\mathbb R^+$ to denote the positive natural (resp. real) numbers.
For a $k \in \mathbb N^+$ we use $[k] \coloneq \{1, \dots, k\}$.
A \emph{unit fraction} is a rational number that can be written as $\frac{1}{k}$ for a $k \in \mathbb N^+$.
Unless otherwise stated, we use $\delta$ to refer to a positive \emph{rational} number.

\subparagraph{Graph Theory}
We consider simple undirected graphs whose edges have unit length.
For a graph $G$, we denote the vertex set with $V(G)$ and the edge set with $E(G)$.
For a set $U \subseteq V(G)$ we denote with $G[U]$ the subgraph of $G$ induced by the vertices of $U$.
Further we denote with $G \setminus U$ the graph $G[V(G) \setminus U]$, and if $U$ contains only one element $u$, we may instead write $G \setminus u$.
For a graph $G$ and integer $\ell \in \mathbb N$, we denote with $G_\ell$ the $\ell$-subdivision of $G$, that is, every edge in $G$ is replaced by a path with $\ell$ edges.

We use the notation of~\cite{DBLP:journals/algorithmica/GrigorievHLW21,DBLP:conf/waoa/HartmannJ24} to describe the metric space of a graph:
For an edge $\{u, v\}$ and a rational $\lambda \in [0,1]$, we denote by $p(u, v, \lambda)$ the point on the edge $\{u, v\}$ that has distance $\lambda$ to $u$.
Thus $p(u, v, \lambda) = p(v, u, 1 - \lambda)$, and we will often assume that $\lambda \leq 1/2$.
Further, $p(u, v, 0) = u$ and $p(u, v, 1) = v$ and we also use $u$ to refer to the point on $u$.
We denote the continuum set of all points on the graph $G$ with $P(G)$.
With the word \emph{point} we refer to the elements of $P(G)$, while we use \emph{vertex} in the graph-theoretic sense.
The distance between two points $p, q \in P(G)$, denoted as $dist(p,q)$, is defined as the length of the shortest path from $p$ to $q$ in the underlying metric space.
The \emph{open ball} with radius $r$ around point $p$, denoted as \oball{p}{r}, is the set of points $q$ with distance $dist(q,p) < r$ from $p$.
The \emph{closed ball} with radius $r$ around point $p$, denoted as \cball{p}{r}, is the set of points $q$ with distance $dist(q,p) \leq r$ from $p$.

Additionally, for a point $p = p(u,v,\lambda)$ we say that $p$ is \emph{on} the edge $\{u,v\}$ or that $\{u,v\}$ \emph{contains} $p$.
In particular, a point $u$ for $u \in V(G)$ is on every edge incident to $u$.
Further, for a point $p = p(u,v,\lambda)$ with $0 < \lambda < 1$, we say that $p$ is \emph{on the interior} of the edge $\{u, v\}$.

\subparagraph{Approximation Preserving Reductions}
Depending on the value of $\delta$, \disp{} resides in the classes \APX{} or \pAPX.
\APX{} are those \NP{} optimization problems that admit a constant-factor approximation in polynomial time, while problems in \pAPX{} only allow for an approximation factor bounded by a polynomial of the input size.

To show our hardness results, we use two types of reductions: L-reductions and S-reductions, as defined in~\cite{crescenzi1997short}.
Both (in our case) are a reduction from a maximization problem $\Pi_1$ to maximization problem $\Pi_2$ given by a pair of functions $f, g$ with the following properties:
\begin{itemize}
	\item The functions $f$ and $g$ are computable in polynomial time.
	\item If $I_1$ is an instance of $\Pi_1$, then $f(I_1)$ is an instance of $\Pi_2$.
	\item If $S_2$ is a solution to an instance $f(I_1)$ of $\Pi_2$, then $g(I_1, S_2)$ is a solution to $I_1$.
\end{itemize}
We use $\OPT_{\Pi}(\cdot)$ to refer to the gain of the optimal solution for $\Pi$ of a given instance, and $\text{gain}_{\Pi}(\cdot)$ to refer to the gain for $\Pi$ of a given solution.
An L-reduction additionally fulfills the following two properties:
\begin{itemize}
	\item There is a constant $\alpha \in \mathbb R^+$ with $\OPT_2(f(I_1)) \leq \alpha \cdot \OPT_1(I_1)$ for all instances $I_1$ of $\Pi_1$.
	\item There is a constant $\beta \in \mathbb R^+$ satisfying $|\text{gain}_1(g(I_1, S_2)) - \OPT_1(I_1)| \leq \beta \cdot |\text{gain}_2(S_2) - \OPT_2(f(I_1))|$ for all instances $I_1$ of $\Pi_1$ and for all solutions $S_2$ of $f(I_1)$.
\end{itemize}
If $\Pi_2$ admits an approximation ratio of $1 + c$, then $\Pi_1$ admits an approximation ratio of $1 + \alpha\beta 
c$ by transformation of the above inequalities.
On the other hand, an S-reduction additionally fulfills the following properties:
\begin{itemize}
	\item For every instance $I_1$ of $\Pi_1$, it holds that $\OPT_2(f(I_1)) = \OPT_1(I_1)$.
    \item For every instance $I_1$ of $\Pi_1$ and every solution $S_2$ of $f(I_1)$, it holds that $\text{gain}(I_1,g(I_1, S_2)) = \text{gain}(f(I_1), S_2)$.
\end{itemize}
For a more in-depth introduction to approximation preserving reductions, we refer the reader to a survey by Crescenzi \cite{crescenzi1997short}.

\subparagraph{Relevant Previous Results}
To obtain a $c$-approximation of \dispn, it suffices to compute a $c$-approximation for every connected component of $G$.
Further, Hartmann et al.~\cite{DBLP:conf/mfcs/HartmannL22} showed that \disp{} admits an \XP-algorithm and is thus solvable in polynomial time on trees.
Therefore, it suffices to consider connected non-tree graphs as input.

\begin{theorem}[Grigoriev et al.~\cite{DBLP:journals/algorithmica/GrigorievHLW21}]
    \label{lemma:polyTimeUnitFraction}
    Let $\delta$ be a unit fraction and $G$ not a tree.
    Then $\dispn = k \cdot |E(G)|$ and an optimal $\delta$-dispersed set $S^*$ contains exactly $k$ points of every edge.
\end{theorem}

\begin{theorem}[Grigoriev et al.~\cite{DBLP:journals/algorithmica/GrigorievHLW21}]
    \label{lemma:polyTimeTwiceUnitFraction}
    Let $\delta = \frac{2}{k}$ for some $k \in \mathbb N$.
    Then \dispn{} can be computed in polynomial time for any graph $G$.
\end{theorem}

Next we show that two self-reductions for \disp{} introduced by Grigoriev et al.~\cite{DBLP:journals/algorithmica/GrigorievHLW21} and Hartmann et al.~\cite{DBLP:conf/mfcs/HartmannL22} respectively are approximation preserving.

\begin{definition}[Subdivision Reduction]
    \label{def:subdivisionReduction}
    This reduction takes $x \in \mathbb N^+$ as a parameter.
    We define $f_x^S(G) = G_x$.
    For an $(x\delta)$-dispersed set $S'$ in $G$ we define $g_x^S(G, S') = S$ as the set containing 
    \begin{itemize}
        \item every point in $S' \cap V(G)$, and
        \item for every point $p' \in S' \setminus V(G)$, the point $p = p(u, v, \frac{\lambda}{x}) \in S$ where $u, v$ are the two closest vertices in $V(G_x) \cap V(G)$ and $\lambda$ is the distance from $u$ to $p'$ in $G_x$.
    \end{itemize}
\end{definition}

The graph $G$ has an optimal $\delta$-dispersed set of size $k$ if and only if the graph $G_x$ has an optimal $(x\delta)$-dispersed set of size $k$ \cite{DBLP:journals/algorithmica/GrigorievHLW21}.
The proof from \cite{DBLP:journals/algorithmica/GrigorievHLW21} immediately yields the following result:

\begin{lemma}
    \label{lemma:subdivisionReduction}
    The Subdivision Reduction is an $L$-reduction for $\alpha = \beta = 1$.
\end{lemma}

The second self-reduction due to Hartmann et al.~\cite{DBLP:conf/mfcs/HartmannL22} transforms $\delta$ such that a $\delta'$-dispersed set places exactly one additional point on every edge.
Technically, the following reduction is only properly defined for $\delta$-{\sc AutoDispersion}.
For a formal definition, we refer the reader to \cite{DBLP:conf/mfcs/HartmannL22}.
$\delta$-{\sc AutoDispersion} is a variant of \disp{} that additionally requires that every point, intuitively speaking, is not too close to itself.
For example for $\delta > 3$, placing a point $p$ on a cycle of length $3$ means that there is a locally injective\footnote{Intuitively speaking, a \emph{locally injective walk} is a walk where at least one edge is only traversed once.} walk from $p$ to itself of length less than $\delta$.
Luckily, for every $\delta \leq 3$, any $\delta$-dispersed set is also $\delta$-auto-dispersed, and additionally it holds that $\autodispn \leq \dispn$, so the following results are still useful to us.

\begin{definition}[Translation Reduction]
    \label{def:translationReduction}
    We define $f^T(G) = G$.
    For a $\frac{\delta}{\delta+1}$-dispersed set $S'$ of $G$, we define $g^T(G,S') = S \subseteq P(G)$ as follows:
    For every edge $\{u,v\} \in E(G)$,
    \begin{itemize}
        \item if $S'$ contains exactly one point\footnote{Hartmann et al.~\cite{DBLP:conf/mfcs/HartmannL22} show that we may w.l.o.g.\ assume that $S'$ contains at least one point from the interior of every edge, since $\frac{\delta}{\delta+1} < 1$.} from the interior of $\{u,v\}$, then $S$ contains no point from the interior of $\{u,v\}$; and 
        \item if $S'$ contains the points $\{p_1',\dots,p_k'\}$ from the interior of $\{u,v\}$ for a $k \geq 2$, then $S$ contains the $k-1$ points $\{p_1,\dots,p_{k-1}\}$ where $p_i = (u,v,\lambda_i)$ and $\lambda_1 = \mu \cdot (\delta+1)$ and $\lambda_{i+1} = \lambda_i + \delta$ for $i \in \{2, \dots, k-2\}$ and $\mu = \min_{i \in [k]}d(p_i,u)$.
    \end{itemize}
    Further, every point in $S'$ that is on a vertex, is also in $S$.
\end{definition}

\begin{theorem}[Hartmann et al.~\cite{DBLP:conf/mfcs/HartmannL22}]
    \label{theorem:translation}
    The function $f^T$ is a polynomial-time reduction from \autodisp{} to \autodisp[\frac{\delta}{\delta+1}] such that $g^T(G,S')$ is a $\delta$-auto-dispersed set in $G$ for every $\frac{\delta}{\delta+1}$-auto-dispersed set $S'$ in $G$.
\end{theorem}

As shown by Hartmann et al.~\cite{DBLP:conf/mfcs/HartmannL22}, $\autodispn = k$ if and only if $\autodispn[\frac{\delta}{\delta+1}] = k + |E(G)|$.
Thus the function $g^T$ exactly translates the absolute error of any approximate solution and we obtain the following result:

\begin{lemma}
    \label{lemma:translationReduction}
    The Translation Reduction $(f^T, g^T)$ is an L-reduction for graphs where $\dispn = c \cdot |E(G)|$ for some constant $c > 0$ with $\alpha = \frac{1}{c}$ and $\beta = 1$.
\end{lemma}
\section{Approximation Algorithms}
\label{sec:approximations}

In this section, we give various constant-factor approximations for \disp{} for $\delta < 2$.
\disp[2] is solvable in polynomial time due to \Cref{lemma:polyTimeTwiceUnitFraction}.
For $\delta > 2$, \disp{} becomes \pAPX-hard, as we show in \Cref{sec:hardness}.

\Cref{lemma:polyTimeUnitFraction,lemma:polyTimeTwiceUnitFraction} yield an infinite number of bounded open intervals for $\delta < 2$ where \disp{} is \NP-hard but computable in polynomial time for either endpoint.
This gives rise to the natural approximation of simply using the solution for the right interval endpoint, as a $\delta$-dispersed set is clearly also $\delta'$-dispersed for any $\delta' < \delta$.

We give approximations for the largest two such intervals, $(1,2)$ and $(\frac{2}{3},1)$.
Throughout this section, we also use this slightly weaker formulation of \Cref{theorem:translation}.

\begin{theorem}[Hartmann et al.~\cite{DBLP:conf/mfcs/HartmannL22}]
    \label{theorem:translationOriginal}
    For any graph $G$ and $\delta > 0$, $\autodispn + |E(G)| = \autodispn[\frac{\delta}{\delta+1}]$.
\end{theorem}

This theorem also allows us to translate our results for $\delta > \frac{2}{3}$ to all smaller intervals.

\subsection{Approximations for $1 < \delta < 2$}

We start with a bound on the optimal solution for $\delta > 1$.
Then we use that result to give an approximation scheme for $\delta \leq \frac{3}{2}$, with an approximation factor approaching $1$ as $\delta \rightarrow 1$.
Finally, we show two approximations for $\delta > \frac{3}{2}$, where the second one again improves the approximation factor towards $1$ as $\delta \rightarrow 2$.

\begin{lemma}
    \label{lemma:maximumSolutionSizeDeltaLargerOne}
    For any graph $G$ and $\delta > 1$, it holds that $\dispn \leq |V(G)|-1$.
\end{lemma}
\begin{proof}
    Let $S$ be some optimal $(1+\varepsilon)$-dispersed set.
    Since $\delta > 1$, no edge with a point from $S$ can be incident to a vertex with a point from $S$, and no two vertices with a point from $S$ can be adjacent.
    Further, the graph induced by the edges that contain a point from $S$ in their interior cannot contain a cycle.
    For a contradiction, assume such a cycle exists, and let $p_1, \dots, p_k$ be the points in the order they appear on the cycle, and let $v_1, \dots, v_k$ be the vertices of the cycle such that $p_i = p(v_i, v_{i+1}, \lambda_i)$ for $i \in [k-1]$ and $p_k = p(v_k, v_1, \lambda_k)$.
    Then $\lambda_{i+1} \geq \lambda_i + \varepsilon$ for $i \in [k-1]$ as otherwise $p_i$ and $p_{i+1}$ have distance less than $\delta$ to each other.
    Then $\lambda_k \geq \lambda_1 + k\varepsilon$ and thus $p_1$ and $p_k$ have distance less than $\delta$ to each other.
    Contradiction to $S$ being $(1+\varepsilon)$-dispersed.
    Let $G'$ be the graph containing exactly the vertices and edges that contain a point from $S$.
    Then $G'$ consists of trees and single vertices.
    Since the single vertices cannot be adjacent in $G$, there may be at most $|V(G)|-1$ of them.
    Each tree $T$ in $G'$ contributes $|V(T)|-1$ points to $S$, and each single vertex contributes $1$.
    Thus in total $S$ can have size at most $|V(G)|-1$.
\end{proof}

With this result, we give our approximation for $\delta \rightarrow 1$ by constructing a heuristic solution that improves in steps as $\delta$ approaches $1$.
As $\delta$ approaches $1$, the approximation factor also approaches $1$.

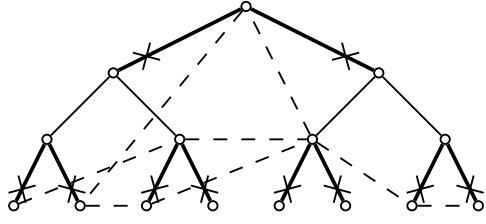
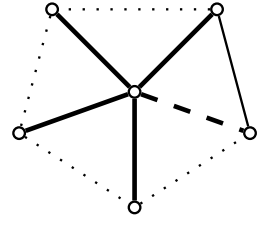
\begin{figure}[t]
	\begin{subfigure}{0.48\textwidth}
		\centering
		\resizebox{\textwidth}{!}{
			\begin{tikzpicture}[scale=0.5,
				node/.style = {shape=circle, draw, inner sep=0pt, minimum size=0.25cm},
				textnode/.style = {shape=circle, draw, inner sep=0pt, minimum size=0.4cm},
				smallnode/.style = {shape=circle, draw, inner sep=0pt, minimum size=0.07cm},
				box/.style = {rectangle, fill=gray!20, rounded corners, fill opacity=1, inner sep=1pt},
				cross/.style={cross out, draw=black, minimum size=0.15cm, inner sep=0pt, outer sep=0pt}]
				% example for algo
				\node[smallnode] (root) at (0, 4) {};
                \node[smallnode] (left) at (-2, 3) {};
                \node[smallnode] (right) at (2, 3) {};
                \node[smallnode] (leftleft) at (-3, 2) {};
                \node[smallnode] (leftright) at (-1, 2) {};
                \node[smallnode] (rightleft) at (1, 2) {};
                \node[smallnode] (rightright) at (3, 2) {};
                \node[smallnode] (leftleftleft) at (-3.5, 1) {};
                \node[smallnode] (leftleftright) at (-2.5, 1) {};
                \node[smallnode] (leftrightleft) at (-1.5, 1) {};
                \node[smallnode] (leftrightright) at (-0.5, 1) {};
                \node[smallnode] (rightleftleft) at (0.5, 1) {};
                \node[smallnode] (rightleftright) at (1.5, 1) {};
                \node[smallnode] (rightrightleft) at (2.5, 1) {};
                \node[smallnode] (rightrightright) at (3.5, 1) {};
			     
                % subtrees
                \draw[thick] (root) -- (left) (root) -- (right);
                \draw[thick] (leftleft) -- (leftleftleft) (leftleft) -- (leftleftright);
                \draw[thick] (leftright) -- (leftrightleft) (leftright) -- (leftrightright);
                \draw[thick] (rightleft) -- (rightleftleft) (rightleft) -- (rightleftright);
                \draw[thick] (rightright) -- (rightrightleft) (rightright) -- (rightrightright);
                % remaining tree
                \draw (left) -- (leftleft) (left) -- (leftright);
                \draw (right) -- (rightleft) (right) -- (rightright);
                % some extra edges
                \draw[dashed] (leftleftleft) -- (leftright) -- (rightleft) -- (root) -- (leftleftright) -- (leftrightleft) -- (rightleft) -- (rightrightleft) -- (rightrightright);
				% points
				\node[cross, rotate=30] at (-1.5, 3.25) {};
				\node[cross, rotate=330] at (1.5, 3.25) {};
				\node[cross, rotate=60] at (-3.375, 1.25) {};
				\node[cross, rotate=300] at (-2.625, 1.25) {};
                \node[cross, rotate=60] at (-1.375, 1.25) {};
				\node[cross, rotate=300] at (-0.625, 1.25) {};
                \node[cross, rotate=60] at (0.625, 1.25) {};
				\node[cross, rotate=300] at (1.375, 1.25) {};
                \node[cross, rotate=60] at (2.625, 1.25) {};
				\node[cross, rotate=300] at (3.375, 1.25) {};
			\end{tikzpicture}
		}
		\caption{Example for the algorithm of \Cref{lemma:approxLargerOneApproachingOne} for $k = 1$. Thick edges are the subtrees, dashed edges are additional edges of the graph. Points placed are marked by crosses.}
		\label{fig:approxSmallerThreeHalvesExample}
	\end{subfigure}
	\hfill
	\begin{subfigure}{0.48\textwidth}
		\centering
		\resizebox{!}{2.8cm}{
			\begin{tikzpicture}[scale=0.5,
				node/.style = {shape=circle, draw, inner sep=0pt, minimum size=0.25cm},
				textnode/.style = {shape=circle, draw, inner sep=0pt, minimum size=0.4cm},
				smallnode/.style = {shape=circle, draw, inner sep=0pt, minimum size=0.07cm},
				box/.style = {rectangle, fill=gray!20, rounded corners, fill opacity=1, inner sep=1pt},
				cross/.style={cross out, draw=black, minimum size=0.15cm, inner sep=0pt, outer sep=0pt}]
				% G_3 with solution for delta = 7/6
				\node[smallnode] (center) at (0, 0) {};
                \node[smallnode] (n1) at (1, 1) {};
                \node[smallnode] (n2) at (-1, 1) {};
                \node[smallnode] (n3) at (0, -1.4) {};
                \node[smallnode] (n4) at (1.4, -0.5) {};
                \node[smallnode] (n5) at (-1.4, -0.5) {};

                \draw[dotted] (n1) -- (n2) -- (n5) -- (n3) -- (n4);
                \draw (n4) -- (n1);
                \draw[thick] (center) -- (n1) (center) -- (n2) (center) -- (n3) (center) -- (n5);
                \draw[thick, dashed] (center) -- (n4);
			\end{tikzpicture}
		}
		\caption{Illustration of the worst case example from \Cref{lemma:approxLargerOneApproachingOne}. The optimal spanning tree uses the thick edges, while the worst choice uses the dotted ones. The thick dashed edge is used by both.}
		\label{fig:approxSmallerThreeHalvesLowerBound}
	\end{subfigure}
	\caption{Examples for the algorithm and lower bound from \Cref{lemma:approxLargerOneApproachingOne}.}
\end{figure}

\begin{lemma}
    \label{lemma:approxLargerOneApproachingOne}
    For $\delta \leq \frac{2k+1}{2k}$ for $k \in \mathbb N^+$, \disp{} can be approximated with a factor of $\frac{k+1}{k}$, for any graph $G$.
    Further, this factor is tight with this approach.
\end{lemma}
\begin{proof}
    First, we compute an arbitrary spanning tree $\mathcal T$ of $G$, and choose the root such that the height of $\mathcal T$ is divisible by $k+1$.
    If no such root exists, $\mathcal T$ has a height of at most $k$, and our procedure will thus place a total $|V(G)|-1$ points, which is optimal by the arguments of \Cref{lemma:maximumSolutionSizeDeltaLargerOne}.
    We say that a leaf of $\mathcal T$ has level $1$, and the level of an inner vertex is the maximum level of its children plus $1$.

    Now we repeat the following procedure until all vertices from $\mathcal T$ are deleted:
    Consider all subtrees $T$ of $\mathcal T$ that are rooted at level $k+1$ vertices.
    For $i \in [k]$, we put the point at distance $\frac{2i-1}{4k}$ from the level $i$ vertex on the edge to its parent in $T$ into $S$.
    Then we delete all subtrees $T$ from $\mathcal T$ and recalculate the levels of the vertices in $\mathcal T$.
    
    Observe that first, $1 - \frac{2k-1}{4k} \geq \frac{\delta}{2}$.
    Thus for any vertex $v \in V(T)$, the points on the edges to the children of $v$ have distance at least $\delta$ to each other.
    Second, $\frac{2x-1}{4k} + 1 - \frac{2(x-1)-1}{4k} = 1 + \frac{2x-1}{4k} - \frac{2x-3}{4k} = 1 + \frac{1}{2k} \geq \delta$.
    Thus for any vertex $v \in V(T)$, the points on the edges to the children of $v$ have distance at least $\delta$ to the point on the edge to the parent of $v$.
    Lastly, $\frac{1}{4k} \geq \frac{\delta-1}{2}$ and thus all points we place have distance at least $\frac{\delta-1}{2}$ to any vertex of $T$.
    Therefore they also have distance at least $\delta$ in $G$, since an edge in $G$ that is not present in $T$ (or connects two trees $T$ and $T'$) creates a path of length at least $1 + 2 \cdot \frac{\delta-1}{2} = \delta$ between any two points in $S$.
    An illustration of this procedure can be seen in \Cref{fig:approxSmallerThreeHalvesExample}.

    In total, for each tree $T$, we place at least $|V(T)|-1$ points in $S$, while an optimal solution could place at most $|V(T)|$.\footnote{Actually, by the arguments of \Cref{lemma:maximumSolutionSizeDeltaLargerOne}, an optimal solution could place at most $|V(G)|-1$ vertices, so we are slightly overestimating the size of the optimal solution. 
    Additionally, there is no guarantee where an optimal solution would place these points, however we just care about the total number.}
    Since $|V(T)| \geq k+1$, we get an approximation factor of at most $\frac{k+1}{k}$ for each tree $T$, and thus also a factor of at most $\frac{k+1}{k}$ in total.

    If the initial tree $\mathcal T$ has height at most $k$, we can instead use the above procedure to place a point on each edge of $\mathcal T$, placing $|V(\mathcal T)| - 1 = |V(G)| - 1$ points, which is optimal by the arguments of \Cref{lemma:maximumSolutionSizeDeltaLargerOne}.

    To see that this factor is tight, consider a wheel graph, which is a cycle with an additional vertex connected to every vertex of the cycle.
    For an example, see \Cref{fig:approxSmallerThreeHalvesLowerBound}.
    One spanning tree of a wheel graph is its Hamiltonian path, in which case our procedure would place $k$ points for every $k+1$ vertices.
    Another spanning tree is a star, which has height $2 \leq k+1$ and thus our procedure would find the optimal solution with $|V(G)| - 1$ points. 
    Thus this example exactly reaches the claimed factor for $|V(G)| \rightarrow \infty$.
\end{proof}

\Cref{lemma:approxLargerOneApproachingOne} covers the values for $\delta \in (1, \frac{3}{2}]$.
Next we show that a $\delta$-dispersed set is at most twice as large as a $2$-dispersed set for $\delta > \frac{3}{2}$, and thus a $2$-dispersed set, which can be computed in polynomial time, is a $2$-approximation.

\begin{figure}[t]
	\begin{subfigure}{0.48\textwidth}
		\centering
		\resizebox{\textwidth}{!}{
			\begin{tikzpicture}[scale=0.5,
				node/.style = {shape=circle, draw, inner sep=0pt, minimum size=0.15cm},
				textnode/.style = {shape=circle, draw, inner sep=0pt, minimum size=0.4cm},
				smallnode/.style = {shape=circle, draw, inner sep=0pt, minimum size=0.07cm},
				box/.style = {rectangle, fill=gray!20, rounded corners, fill opacity=1, inner sep=1pt},
				cross/.style={cross out, draw=black, minimum size=0.15cm, inner sep=0pt, outer sep=0pt}]
				% example for algo
				\node[node] (left) at (-4, 0) {};
                \node[inner sep=0pt, minimum size=0cm] (leftleft) at (-3, 0) {};
                \node[inner sep=0pt, minimum size=0cm] (leftmiddle) at (-1, 0) {};
                \node[node] (middle) at (0, 0) {};
                \node[inner sep=0pt, minimum size=0cm] (rightmiddle) at (1, 0) {};
                \node[inner sep=0pt, minimum size=0cm] (rightright) at (3, 0) {};
                \node[node] (right) at (4, 0) {};
			     
                % edges
                \draw[dashed] (left) -- (leftleft) (leftmiddle) -- (middle) -- (rightmiddle) (rightright) -- (right);
                \draw[ultra thick] (leftleft) -- (leftmiddle) (rightmiddle) -- (rightright);
				% points
				\node[cross,label=below:$p$] at (-2.5, 0) {};
				\node[cross,label=above:$q$] at (3.2, 0) {};
                % areas
                \draw [decorate,decoration={brace,mirror,amplitude=5pt,raise=2ex}] (-4, 0) -- (leftleft) node[midway,yshift=-2em]{$S_V$};
                \draw [decorate,decoration={brace,mirror,amplitude=5pt,raise=2ex}] (leftmiddle) -- (rightmiddle) node[midway,yshift=-2em]{$S_V$};
                \draw [decorate,decoration={brace,mirror,amplitude=5pt,raise=2ex}] (rightright) -- (4, 0) node[midway,yshift=-2em]{$S_V$};
                \draw [decorate,decoration={brace,amplitude=5pt,raise=2ex}] (leftleft) -- (leftmiddle) node[midway,yshift=2em]{$S_E$};
                \draw [decorate,decoration={brace,amplitude=5pt,raise=2ex}] (rightmiddle) -- (rightright) node[midway,yshift=2em]{$S_E$};
			\end{tikzpicture}
		}
		\caption{Example for the algorithm of \Cref{lemma:approxLargerThreeHalves}. Points on thick edge parts (here $p$) are moved to midpoints in $S_E$, while the others are moved to vertices in $S_V$ (here $q$).}
		\label{fig:approxLargerThreeHalvesExample}
	\end{subfigure}
	\hfill
	\begin{subfigure}{0.48\textwidth}
		\centering
		\resizebox{\textwidth}{!}{
			\begin{tikzpicture}[scale=0.5,
				node/.style = {shape=circle, draw, inner sep=0pt, minimum size=0.25cm},
				textnode/.style = {shape=circle, draw, inner sep=0pt, minimum size=0.4cm},
				smallnode/.style = {shape=circle, draw, inner sep=0pt, minimum size=0.07cm},
				box/.style = {rectangle, fill=gray!20, rounded corners, fill opacity=1, inner sep=1pt},
				cross/.style={cross out, draw=black, minimum size=0.15cm, inner sep=0pt, outer sep=0pt}]
				% triangles
				\node[smallnode] (center) at (0, 0) {};
                \node[smallnode] (t11) at (2, -1) {};
                \node[smallnode] (t12) at (2.5, -2) {};
                \node[smallnode] (t13) at (1.5, -2) {};
                \node[smallnode] (t21) at (0, -1) {};
                \node[smallnode] (t22) at (0.5, -2) {};
                \node[smallnode] (t23) at (-0.5, -2) {};
                \node[smallnode] (t31) at (-2, -1) {};
                \node[smallnode] (t32) at (-2.5, -2) {};
                \node[smallnode] (t33) at (-1.5, -2) {};

                % points 
                \node[cross] at (2, -2) {};
                \node[cross, rotate=330] at (1.667, -0.833) {};
                \node[cross] at (0, -2) {};
                \node[cross] at (0, -0.833) {};
                \node[cross] at (-2, -2) {};
                \node[cross, rotate=30] at (-1.667, -0.833) {};

                % edges
                \draw (center) -- (t11) (center) -- (t21) (center) -- (t31);
                \draw (t11) -- (t12) -- (t13) -- (t11);
                \draw (t21) -- (t22) -- (t23) -- (t21);
                \draw (t31) -- (t32) -- (t33) -- (t31);
			\end{tikzpicture}
		}
		\caption{Illustration of the worst case example from \Cref{lemma:approxLargerThreeHalves}. A $\delta$-dispersed set can place points as shown, while a $2$-dispersed set only places one point per triangle and one on the central vertex.}
		\label{fig:approxLargerThreeHalvesLowerBound}
	\end{subfigure}
	\caption{Examples for the algorithm and lower bound from \Cref{lemma:approxLargerThreeHalves}.}
\end{figure}
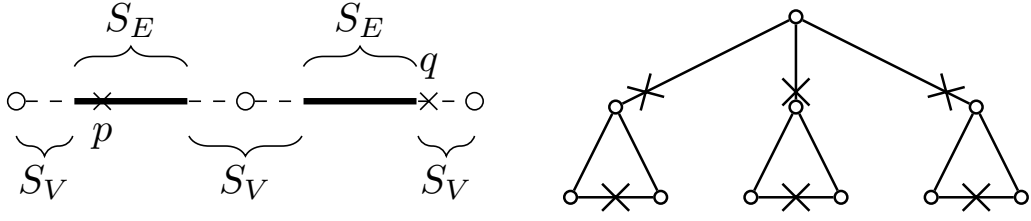

\begin{lemma}
    \label{lemma:approxLargerThreeHalves}
    For any graph $G$ and $\delta > \frac{3}{2}$, it holds that $\half \cdot \dispn \geq \dispn[2]$.
    Additionally for $\delta \leq \frac{5}{3}$, this factor is tight.
\end{lemma}
\begin{proof}
    We split an arbitrary $\delta$-dispersed set $S$ into two sets $S_V$ and $S_E$ which are both $2$-dispersed:
    \begin{itemize}
        \item Let $S_V = \{v ~\mid~ \oball{v}{\frac{1}{4}} \cap S \neq \emptyset\}$ and
        \item let $S_E = \{p = p(u,v,\half) ~\mid~ \cball{p}{\frac{1}{4}} \cap S \neq \emptyset\}$.
    \end{itemize}
    An example of this can be seen in \Cref{fig:approxLargerThreeHalvesExample}.
    Assume $S_V$ is not $2$-dispersed.
    Then there are two adjacent vertices $u,v \in S_V$, as $S_V$ only contains points on vertices by definition.
    Let $p_u$ and $p_v$ be the points in $\oball{u}{\frac{1}{4}}$ (resp. $\oball{v}{\frac{1}{4}}$).
    Since $\{u, v\} \in E(G)$, $dist(p_u,p_v) < \frac{3}{2} < \delta$, which is a contradiction to $S$ being $\delta$-dispersed.
    Now assume $S_E$ is not $2$-dispersed.
    Then there are two incident edges $\{u,v\}, \{v,w\}$ with $p = p(u,v,\half), p' = p(u',v',\half) \in S_E$.
    Let $q,q'$ be the points in $\cball{p}{\frac{1}{4}}$ and $\cball{p'}{\frac{1}{4}}$, respectively.
    Since $\{u,v\}, \{v,w\}$ share a vertex, $dist(q,q') \leq \frac{3}{2} < \delta$, which again is a contradiction to $S$ being $\delta$-dispersed.

    Since $|S_V| + |S_E| = |S|$, it holds that an optimal $2$-dispersed set must have at least half the size of an optimal $\delta$-dispersed set.

    To see that this factor is tight, consider a graph consisting of $\ell$ triangles connected to a central vertex.
    For an example, see \Cref{fig:approxLargerThreeHalvesLowerBound}.
    For $\delta \in (\frac{3}{2}, \frac{5}{3}]$, an optimal solution can place exactly two points on each triangle gadget:
    One on the midpoint of the triangle edge opposite to the vertex connected to the central vertex, and one with distance $\frac{1}{6}$ from the vertex connected to the central vertex on the edge to the central vertex.
    To see that it is not possible to place more than two points per triangle gadget, note that for $\delta > \frac{3}{2}$ no more than one point can be placed on a triangle.
    For $\delta = 2$, there can be at most one point on the edges connecting the triangle gadgets to the central vertex.
    Additionally, there can still be only one point per triangle, resulting in a total number of $\ell + 1$ points.
    Thus we get a ratio of $\frac{2\ell}{\ell+1}$, which approaches $2$ as $\ell \rightarrow \infty$.
\end{proof}

Now we improve this approximation for $\delta \rightarrow 2$.
Instead of giving a heuristic solution that improves as $\delta$ approaches $2$, we show that a $\delta$-dispersed set is not much larger than a $2$-dispersed set, showing that using a $2$-dispersed set is a better approximation than suggested by \Cref{lemma:approxLargerThreeHalves} as $\delta \rightarrow 2$.

The idea is to split a $\delta$-dispersed set into disjoint subsets, such that removing one of them allows us to move the remaining points such that they are $2$-dispersed.
By increasing the number of sets as $\delta \rightarrow 2$, the approximation factor approaches $1$.

\begin{lemma}
    \label{lemma:approxLargerOneApproachingTwo}
    For any graph $G$ and $\delta > \frac{4k+1}{2k+1}$ for $k \in \mathbb N^+$, $\frac{k}{k+1} \cdot \dispn \geq \dispn[2]$.
\end{lemma}
\begin{proof}
    We consider some optimal $\delta$-dispersed set $S$, and partition the points of $S$ into $k+1$ sets:
    $S^*$, $S_i$ for $i \in [k-1]$ and $S_{mid}$.
    All points from $S$ that are on a vertex or on an edge such that there is an incident edge that also contains a point from $S$ are contained in $S^*$.
    Let $G^*$ be the graph induced by the vertices and edges that contain a point from $S^*$.
    We say a path starting in a vertex $v$ from $G^*$ with $\cball{v}{\half} \cap S \neq \emptyset$ is \emph{alternating}, if the first edge does not contain a point from $S$, and afterwards every other edge contains a point from $S$.
    For a point $p \in S \setminus S^*$ let $e_p$ be the edge containing $p$.
    Then $p \in S_i$ if and only if the shortest alternating path from any vertex in $G^*$ reaching either endpoint of $e_p$ contains exactly $i$ edges that do not contain a point from $S$.
    Finally, $S_{mid} = S \setminus (S^* \cup \bigcup_{i \in [k-1]}S_i)$.

    Let $v$ be a vertex in $G^*$.
    If $\cball{v}{\half} \cap S \neq \emptyset$ then also $\cball{v}{2-\delta} \cap S \neq \emptyset$ (note that $2 - \delta < \half$).
    Assume this is not the case.
    Then $v$ is not isolated in $G^*$, as otherwise there would be a point on $v$.
    Let $e = \{v, w\}$ be the edge where the $p \in \cball{v}{\half} \cap S$ is located, and let $p = p(v, w, \lambda)$.
    Since $\delta > \frac{3}{2}$, there is no point from $S$ on the other edges incident to $v$.
    If $\lambda > 2-\delta$, any point on an edge incident to $w$ has distance strictly less than $1 + (1 - (2 - \delta)) = \delta$ to $p$ and therefore cannot be contained in $S$, which is a contradiction to the construction of $S^*$.
    Thus, $\lambda \leq 2-\delta$ must hold.
    
    For any $p \in S_i$ let $e_p = \{u_p, v_p\}$ be the edge containing $p$, and w.l.o.g.\ let $u_p$ be the vertex that is reached by the shortest alternating path from a vertex in $G^*$.
    Then $p = p(u_p, v_p, \lambda_p)$ with $\lambda_p \geq 1- (i+1) \cdot (2-\delta)$.
    Assume that this is not the case, i.e. $\lambda_p < 1- (i+1) \cdot (2-\delta)$.
    Consider the sequence of points $q_1, \dots, q_{i+1} \in S$ along the shortest alternating path from vertex $v \in V(G^*)$ to $u_p$ with $q_1 \in \cball{v}{2-\delta}$ and $q_{i+1} = p$.
    By definition of $S_i$ and the arguments above, such a sequence must exist.
    Then 
    %$d(p, q_1) < 2i - 1 + (2-\delta) + 1 - (i+1) \cdot (2-\delta) = 2i - i \cdot (2-\delta) = i\cdot\delta$ 
    \begin{equation*}
        d(p, q_1) < 2i - 1 + (2-\delta) + 1 - (i+1) \cdot (2-\delta) = 2i - i \cdot (2-\delta) = i\cdot\delta
    \end{equation*}
    and thus by the pigeon-hole principle there are two consecutive points in the sequence $q_1, \dots, q_{i+1}$ with distance less than $\delta$ from each other, which is a contradiction to $S$ being $\delta$-dispersed.

    From $\lambda_p \geq 1- (i+1) \cdot (2-\delta) > \half$ for any $k \in \mathbb N^+$ it follows that the choice of $u_p$ above is unique, as there cannot be an alternating path of the same length reaching $v_p$.
    Thus, for any vertex $v_p$ for $p \in \bigcup_{i \in [k-1]}S_i$ and $w \in V(G^*)$ with $\cball{w}{\half} \cap S \neq \emptyset$, the edge $\{v_p, w\} \notin E(G)$.
    Further, it holds that $1 - j \cdot (2-\delta) > 1- (j+1)\cdot(2-\delta)$, and by choice of $\delta$, 
    %$1 - k\cdot (2-\delta) > 1 - 2k + k\cdot\frac{4k+1}{2k+1} = 1 - k + \frac{2k^2}{2k+1} = 1 - k + \frac{2\cdot(k+1)(k-1) + 2}{2k+1} = 1 - k + (k-1)\frac{2k + 2}{2k+1} + \frac{2}{2k+1} = 1-k + (k-1)(1 + \frac{1}{2k+1}) + \frac{2}{2k+1} = \frac{k+1}{2k+1} = 1 - \frac{k}{2k+1} > 1 - \frac{\delta-1}{2}$.
    \begin{align*}
        1 - k\cdot (2-\delta) &> 1 - 2k + k\cdot\frac{4k+1}{2k+1} = 1 - k + \frac{2k^2}{2k+1} = 1 - k + \frac{2\cdot(k+1)(k-1) + 2}{2k+1} \\
        &= 1 - k + (k-1)\frac{2k + 2}{2k+1} + \frac{2}{2k+1} = 1-k + (k-1)(1 + \frac{1}{2k+1}) + \frac{2}{2k+1} \\
        &= \frac{k+1}{2k+1} = 1 - \frac{k}{2k+1} > 1 - \frac{\delta-1}{2}.
    \end{align*}
    Therefore, for any two vertices $v_p, v_q$ for $p, q \in \bigcup_{i \in [k-1]}S_i$ as defined above, $\{v_p, v_q\} \notin E(G)$, as otherwise $d(p,q) < \delta$.
    Then by definition of the sets $S^*$, $S_i$ for $i \in [k-1]$ and $S_{mid}$, the following sets are $2$-dispersed:
    \begin{enumerate}
        \item $\mathcal S_0 = \{p(u,v,0) ~|~ \oball{u}{\half} \cap (S \setminus S_{mid}) \neq \emptyset\}$, i.e. we move all points to their nearest vertex (note that by the arguments above, no point in $S \setminus S_{mid}$ is on the midpoint of an edge), and delete the points in $S_{mid}$.
        \item For $i \in [k-1]$ we set 
        %$\mathcal S_i = \{p(u,v,0) ~|~ \oball{u}{\half} \cap (S^* \cup \bigcup_{j \in [i-1]} S_j)\} \cup \{p(u,v,\half) ~|~ p(u,v,\lambda) \in S_{mid} \cup \bigcup_{j \in [k-1]\setminus[i]} S_j\}$,
        \begin{align*}
            \mathcal S_i &= \{p(u,v,0) ~|~ \oball{u}{\half} \cap (S^* \cup \bigcup_{j \in [i-1]} S_j)\} \\
            &\phantom{{}={}}\cup \{p(u,v,\half) ~|~ p(u,v,\lambda) \in S_{mid} \cup \bigcup_{j \in [k-1]\setminus[i]} S_j\},
        \end{align*}
        i.e. we delete the set $S_i$, and move all points in $S^*$ and $S_j$ with $j < i$ to their nearest vertex, while we move remaining points to the midpoint of their edge.
        \item $\mathcal S_{k+1} = \{p(u,v,\half) ~|~ p(u,v,\lambda) \in S \setminus S^*\}$, i.e. we delete all points from $S^*$ and move the remaining points to the midpoint of their edge.
    \end{enumerate}
    By the arguments above, all vertices where a point can get moved to are an independent set in $G$, and none of the edges on whose midpoints a point can get placed share a vertex.
    Therefore, the sets $\mathcal S_0$ and $\mathcal S_{k+1}$ are $2$-dispersed.
    For the sets $\mathcal S_i$, the points placed on vertices and midpoints of edges are $2$-dispersed when considered separately by the same argument.
    Additionally, by the definition of $S_i$, the vertices in $\{v ~|~ \cball{v}{\half} \cap (S^* \cup \bigcup_{j \in [i-1]} S_j) \neq \emptyset\}$ and $\{w ~|~ \cball{w}{\half} \cap (S_{mid} \cup \bigcup_{j \in [k-1]\setminus[i]} S_j) \neq \emptyset\}$ have distance at least $3$, since the points in $S_i$ are in the interior of an edge, and there can only be edges without points from $S$ in their interior incident to those edges.
    Therefore, the sets $\mathcal S_i$ are also $2$-dispersed.

    By the pigeon-hole principle, one of the sets $S^*$, $S_i$ for $i \in [k-1]$ and $S_{mid}$ must have size at most $\frac{1}{k+1}\cdot|S|$.
    Then the respective set $\mathcal S_i$ not containing that set has size $\frac{k}{k+1}\cdot|S|$ and is $2$-dispersed.
    Thus the claim follows.
\end{proof}

With this lemma, we showed that computing a $2$-dispersed set, which is also $\delta$-dispersed for $\delta \leq 2$, is a $\frac{k+1}{k}$-approximation for $\delta > \frac{4k+1}{2k+1}$ and $k \in \mathbb N^+$.

\begin{corollary}
    \label{cor:approxLargerOneApproachingTwo}
    For any graph $G$ and $\delta > \frac{4k+1}{2k+1}$ for $k \in \mathbb N^+$, \disp{} can be approximated in polynomial time with a factor of $\frac{k+1}{k}$.
\end{corollary}

\subsection{Approximations for $\frac{2}{3} < \delta < 1$}

Next, we derive an approximation for the interval $(\frac{2}{3}, 1)$.
As $\delta \rightarrow 1$, the approximation factor also approaches $1$, and for $\delta > \frac{2}{3}$ it starts with $\frac{3}{2}$.
In contrast to the previous section, we give no approximation scheme whose approximation factor approaches $1$ as $\delta \rightarrow \frac{2}{3}$, since, as we show in \Cref{sec:hardness}, it is unlikely to exist unless $\PTIME = \NP$.
Instead of giving a heuristic that improves as $\delta \rightarrow 1$, we show that the size of the optimal solution approaches the size of a $1$-dispersed set. 
For that, we give bounds on the maximum size of a $\delta'$-dispersed set as $\delta' \rightarrow \infty$.
We remind the reader that $\autodispn \leq \dispn$ for any graph $G$, and for $\delta \leq 3$ even $\autodispn = \dispn$.
Thus, we may use \Cref{theorem:translationOriginal} to translate these upper bounds to $\delta \in (\frac{2}{3}, 1)$.

\begin{lemma}
    \label{lemma:approxSmallerOne}
    For any graph $G$ and $\delta > \frac{2k}{2k+1}$ for $k \in \mathbb N^+$, an optimal $1$-dispersed set $S$ is a $\frac{k+2}{k+1}$-approximation for \disp. 
    Further, for $\delta \leq \frac{2k+2}{2k+3}$, this bound is tight.
\end{lemma}
\begin{proof}
    We give an upper bound on \dispn[(2k+\varepsilon)] and then use \Cref{theorem:translationOriginal} to obtain an upper bound on \dispn[(\frac{2k}{2k+1}+\varepsilon')] for $\varepsilon, \varepsilon' > 0$.
    Let $S^*$ be a $(2k+\varepsilon)$-dispersed set in $G$.
    For now, we assume that $|S^*| \geq 2$.
    For a point $p \in S^*$, let $V_p = \cball{p}{k} \cap V(G)$.
    Then for two points $p,q \in S^*$, $V_p \cap V_q = \emptyset$.
    Further, $V_p$ contains at least $k+1$ vertices:
    \begin{enumerate}
        \item If $p$ is on an edge $e = \{u, v\}$, then $V_p$ contains its endpoints. 
            Further, since $|S^*| > 1$, there must be a vertex $w \in V(G) \setminus V_p$, and since $G$ is connected, there must be a path from $v$ to $w$.
            Now let $w$ be the vertex in $V(G) \setminus V_p$ with the shortest distance to either endpoint of $e$, and w.l.o.g.\ let that endpoint be $v$.
            Then the path from $v$ to $w$ contains $k+1$ vertices such that $d(v,w) = k$ (as otherwise $d(p, w) \leq k$), and $k$ of them are contained in $V_p$ (all but $w$).
            Together with $u$, $V_p$ thus contains at least $k+1$ vertices.
        \item Otherwise, $p$ is on a vertex $v$.
            Further, since $|S^*| > 1$, there must be a vertex $w \in V(G) \setminus V_p$, and since $G$ is connected, there must be a path from $v$ to $w$.
            Now let $w$ be the vertex in $V(G) \setminus V_p$ with the shortest distance to $v$.
            Then $d(v,w) = k+1$, and thus the shortest path from $v$ to $w$ contains $k+2$ vertices, and all of them but $w$ are contained in $V_p$.
            Thus $V_p$ contains at least $k+1$ vertices.
    \end{enumerate}
    Since for any two points $p,q \in S^*$, $V_p \cap V_q = \emptyset$, it follows that $|S^*| \leq \frac{1}{k+1}|V(G)|$.

    Now consider the case of $|S^*| = 1$.
    If $|V(G)| \geq k+1$, then $|S^*| \leq \frac{1}{k+1}|V(G)|$, and otherwise $G$ has constant size and thus \dispn{} can be computed in constant time by brute force for any rational $\delta$ as \disp{} is in \NP~\cite{DBLP:journals/algorithmica/GrigorievHLW21}.

    Thus, w.l.o.g., $\dispn[(2k+\varepsilon)] \leq \frac{1}{k+1}|V(G)|$, and by \Cref{theorem:translationOriginal} we have $\dispn[(\frac{2k}{2k+1}+\varepsilon')] = |E(G)| + \frac{1}{k+1}|V(G)|$.
    By \Cref{lemma:polyTimeUnitFraction}, $\dispn[1] = |E(G)|$, and thus a $(\frac{2k}{2k+1}+\varepsilon')$-dispersed set can be at most 
    $$
        \frac{|E(G)| + \frac{1}{k+1}|V(G)|}{|E(G)|} \leq 1 + \frac{|V(G)|}{(k+1)|E(G)|} \leq \frac{k+2}{k+1}
    $$
    times larger than a $1$-dispersed set.

    To see that this approximation factor is tight, consider a star subdivided $k+1$ times.
    Then an optimal $\delta$-dispersed set for $2k < \delta \leq 2k+2$ has the same size as the number of leaves of the star, and is also auto-dispersed.
    Then by \Cref{theorem:translation} we have equality in the above equations.\footnote{Only asymptotically, as the number of vertices of a $(k+1)$-times subdivided star is $(k+1)\cdot x + 1$ where $x$ is the number of leaves.}
\end{proof}

\section{Hardness of Approximation}
\label{sec:hardness}

In this section, we complement our approximations from \Cref{sec:approximations} with \APX- and \pAPX-hardness results.
We begin with \pAPX-hardness for all $\delta > 2$, explaining the lack of approximations for that range of $\delta$ in the previous section, and then show \APX-hardness for all remaining $\delta$ that are not covered by \Cref{lemma:polyTimeUnitFraction,lemma:polyTimeTwiceUnitFraction}.
Throughout this section, we also use $n$ for the number of vertices of a given graph.

\subsection{Hardness for $\delta > 2$}

In this section, we show that \disp{} is \pAPX-complete for any $\delta > 2$.
To see that \disp{} is also contained in \pAPX, recall \Cref{lemma:maximumSolutionSizeDeltaLargerOne}, that for any $\delta > 1$ it holds that $\dispn < |V(G)|$ and thus \disp{} is contained in \pAPX.
To show hardness, we give reductions from {\sc Independent Set}, which is known to be hard to approximate with factor $n^{1-\varepsilon}$ for any $\varepsilon > 0$~\cite{zuckermanIndependentSetHardness}, and therefore \pAPX-hard.
In particular, we show hardness for the intervals $(2, 3]$ and $(3, 4]$ using reductions from {\sc Independent Set}, and then use the Subdivision Reduction (see \Cref{def:subdivisionReduction}) to extend the result for $(2, 3]$ to $(2k, 3k]$ for $k \in \mathbb N^+$.
Since $\bigcup_{k \in \mathbb N^+} (2k, 3k] = (2, \infty) \setminus (3, 4]$, these intervals together cover the entire interval $(2, \infty)$.

\begin{lemma}
    \label{lemma:hardnessTwoToThree}
    \disp{} is \NP-hard to approximate within a factor better than $(\frac{n}{2})^{1-\varepsilon}$ for any $\delta \in (2, 3]$ and any $\varepsilon > 0$.
\end{lemma}
\begin{proof}
    We give an $S$-reduction from {\sc Independent Set} to \disp{} for $\delta \in (2, 3]$.
    For two graphs $G_1$ and $G_2$ we denote the \emph{wreath product} as $G_1 \wreath G_2$, where $V(G_1 \wreath G_2) = V(G_1) \times V(G_2)$ and 
    $$
        E(G_1 \wreath G_2) = \{\{(v_1, v_2), (v_1', v_2')\} ~|~ \{v_1, v_1'\} \in E(G_1), \text{ or } v_1 = v_1' \text{ and } \{v_2, v_2'\} \in E(G_2)\}.
    $$
    We then define $f(G) = G \wreath P_2$, where $P_2$ is the path on two vertices, which we name $w_1, w_2$.
    Further, for a $\delta$-dispersed set $S'$ in $f(G)$, we define 
    \begin{align*}
        g(f(G), S') &= \{v ~|~ p((v,w_i),(u,w_i),\lambda) \in S' \text{ with } \lambda < \half, \text{ or } \lambda = \half \text{ and } v < u\} \\
        &\hspace{0.08cm}\cup \{v ~|~ p((v,w_1),(v,w_2),\lambda) \in S'\},
    \end{align*}
    where we assume the existence of an arbitrary total ordering on the vertices of $G$.
    Since $\cball{p((v,w_1),(v,w_2),\half}{1}$ can contain at most one point from $S'$ for any $\delta > 2$, $|g(f(G), S')| = |S'|$ for any $\delta$-dispersed set $S'$ in $f(G)$.

    It remains to show that for any independent set $I$ in $G$ we can define an equal-sized $\delta$-dispersed set in $f(G)$ and that $g(f(G), S')$ is an independent set in $G$ if $S'$ is $\delta$-dispersed in $f(G)$.

    \begin{claim}
        \label{claim:hardnessTwoToThreeClaimA}
        For any independent set $I$ in $G$ we can define an equal-sized $\delta$-dispersed set in $f(G)$ for any $\delta \in (2, 3]$.
    \end{claim}
    \begin{claimproof}
        We define $S' = \{p((v,w_1),(v,w_2),\half) ~|~ v \in I\}$.
        Obviously, $|S'| = |I|$.
        Further, since $I$ is an independent set, for any two vertices $v, v' \in I$ it holds that $d(v,v') \geq 2$.
        Then $d(p((v,w_1),(v,w_2),\half), p((v',w_1),(v',w_2),\half)) \geq 3$.
        Thus $S'$ is $\delta$-dispersed for any $\delta \in (2, 3]$.
    \end{claimproof}

    \begin{claim}
        \label{claim:hardnessTwoToThreeClaimB}
        $g(f(G), S')$ is an independent set in $G$ if $S'$ is $\delta$-dispersed in $f(G)$.
    \end{claim}
    \begin{claimproof}
        Let $I = g(f(G), S')$ and for a contradiction assume that $I$ is not an independent set in $G$.
        Then there are vertices $v_1, v_2 \in I$ with $\{v_1, v_2\} \in E(G)$.
        Let $p_i \in S'$ be the point that caused $v_i$ to get added to $I$.
        Then $p_i$ has distance at most $\half$ from $(v_i, w_1)$ or $(v_i, w_2)$.
        Since $\{v_1, v_2\} \in E(G)$, it thus holds that $d(p_1, p_2) \leq 2 < \delta$, which is a contradiction to $S'$ being $\delta$-dispersed.
    \end{claimproof}

    From Claims \ref{claim:hardnessTwoToThreeClaimA} and \ref{claim:hardnessTwoToThreeClaimB} and $|g(f(G), S')| = |S'|$ it follows that the size of an optimal independent set in $G$ and an optimal $\delta$-dispersed set in $f(G)$ is the same.
    Thus $(f,g)$ is an $S$-reduction.
    Since $f(G)$ has $2|V(G)|$ vertices, we obtain the claimed inapproximability result.
\end{proof}

Next we give a reduction for the interval $(3, 4]$.

\begin{lemma}
    \label{lemma:hardnessThreeToFour}
    There is a $c \in \mathbb R^+$ such that \disp{} is \NP-hard to approximate within any factor better than $c\cdot n^{0.5-\varepsilon}$ for any $\delta \in (3,4]$ and any $\varepsilon > 0$.
\end{lemma}
\begin{proof}
    We give an $S$-reduction from {\sc Independent Set} to \disp{} for $\delta \in (3, 4]$.
    We define $f(G) = G'$, where $G'$ is the $2$-subdivision of $G$, with an additional vertex $v^*$ connected to all new vertices.
    Thus $d(v^*,p) \leq 2$ for any $p \in P(G')$, and therefore $d(v_e,p) \leq 3$ for any vertex $v_e \in V(G') \setminus V(G)$ and $p \in P(G)$.
    Therefore for any $\delta$-dispersed set $S'$ of $G'$ with $|S'| > 1$, $\cball{v^*}{1} \cap S' = \emptyset$.
    For a $\delta$-dispersed set $S'$ in $G'$, we define $g(G', S') = \{v \in V(G) ~|~ \oball{v}{1} \cap S' \neq \emptyset\}$ if $S'$ contains at least $2$ points, and we let $g(G', S')$ contain a single arbitrary vertex of $V(G)$ otherwise.
    Since $\oball{v}{1}$ can contain at most one point from any $\delta$-dispersed set $S'$, it is easy to see that $|g(G', S')| = |S'|$.

    It remains to show that for any independent set $I$ in $G$ we can define an equal-sized $\delta$- dispersed set in $G'$, and that $g(G', S')$ is an independent set in $G$ if $S'$ is $\delta$-dispersed in $G'$.

    \begin{claim}
        \label{claim:hardnessThreeToFourClaimA}
        For any independent set $I$ in $G$ we can define an equal-sized $\delta$-dispersed set in $G'$ for any $\delta \in (3, 4]$.
    \end{claim}
    \begin{claimproof}
        We note that by construction, any two vertices that share an edge in $G$ have distance $2$ in $G' \setminus {v^*}$ and distance at least $4$ otherwise.
        The addition of $v^*$ does not change the distance for vertices that are adjacent in $G$, but for vertices that are not adjacent in $G$ it makes sure that there is always a path of length $4$.
        Thus, we define $S' = I$, and since none of the vertices in $I$ are adjacent in $G$, they mutually have distance at least $4$ in $G'$ and are thus $\delta$-dispersed.
    \end{claimproof}

    \begin{claim}
        \label{claim:hardnessThreeToFourClaimB}
        $g(G', S')$ is an independent set in $G$ if $S'$ is $\delta$-dispersed in $G'$.
    \end{claim}
    \begin{claimproof}
        The case of $|S'| = 1$ is clear, so we only deal with the case $|S'| > 1$.
        For a contradiction, assume that there are two vertices $v, w \in g(G', S')$ that are adjacent.
        Let $e_{vw} \in V(G') \setminus V(G)$ be the vertex connecting them in $G'$.
        Further, let $p_v$ and $p_w$ be the points in $S'$ with $d(v,p_v) < 1$ and $d(w, p_w) < 1$.
        Then $p_v, p_w \notin \cball{e_{vw}}{1}$, as otherwise they would have distance at most $3$.
        Let $e_{vv'} \in V(G') \setminus V(G)$ be the vertex such that $p_v = p(v, e_{vv'}, \lambda_v)$ and respectively $e_{ww'} \in V(G') \setminus V(G)$ such that $p_w = p(w, e_{ww'}, \lambda_w)$.
        Then by construction, $v, e_{vw}, w, e_{ww'}, v^*, e_{vv'}, v$ is a cycle of length $6$, which is a contradiction to $S'$ being $\delta$-dispersed for $\delta > 3$.
    \end{claimproof}

    From Claims \ref{claim:hardnessThreeToFourClaimA} and \ref{claim:hardnessThreeToFourClaimB} and $|g(G', S')| = |S'|$ it follows that the size of an optimal independent set in $G$ and an optimal $\delta$-dispersed set in $G'$ is the same.
    Thus $(f,g)$ is an $S$-reduction.
    Since $G'$ has $\mathcal O(|V(G)|^2)$ vertices, we obtain the claimed inapproximability result.
\end{proof}

Together with the Subdivision Reduction and the above discussion, \Cref{lemma:hardnessTwoToThree,lemma:hardnessThreeToFour} yield the following theorem.

\begin{theorem}
    \label{thm:polyAPXHardness}
    \disp{} is \pAPX-complete for any $\delta > 2$.
\end{theorem}

\subsection{Hardness for $\delta < 2$}

For the hardness for $\delta < 2$, we use the Subdivision Reduction and the Translation Reduction to translate hardness results for $\delta > 2$ to smaller intervals.
As the Translation Reduction is only an L-reduction if the solution size is at least $c \cdot |E(G)|$ for some constant $c \in \mathbb R^+$, we start our reduction chain with {\sc Independent Set} on $k$-regular graphs.
Clearly, for a $k$-regular graph, there is an independent set with at least $\frac{1}{k+1} \cdot |V(G)|$ vertices (for example by repeatedly choosing an arbitrary vertex to add to the independent set, and then deleting all its neighbors, removing at most $k+1$ vertices in each iteration).
Further, a $k$-regular graph has $\frac{k}{2}\cdot|V(G)|$ many edges, and thus the size of an independent set in a $k$-regular graph has always size at least $\frac{2}{k^2+k}\cdot|E(G)|$, a constant fraction of the number of edges.

Since {\sc Independent Set} on $k$-regular graphs is \APX-complete for $k \geq 3$ \cite{ALIMONTI2000123,PAPADIMITRIOU1991425}, \Cref{lemma:hardnessTwoToThree,lemma:hardnessThreeToFour} also yield \APX-hardness for \disp{} on graphs with solution size in $\Omega(|E(G)|)$.
Note that the reduction in \Cref{lemma:hardnessThreeToFour} does not preserve the bounded degree property, however only the lower bound on the size of the optimum solution is relevant here.
By repeatedly applying the Translation Reduction to the interval $(2, \infty)$, we obtain the following lemma.

\begin{lemma}
    \label{lemma:apxHardnessNumeratorTwo}
    For $\delta \in (\frac{2}{2x+1}, \frac{1}{x})$ and $x \in \mathbb N^+$, \disp{} is \APX-hard.
\end{lemma}

As \Cref{lemma:hardnessTwoToThree} uses the same reduction for the entire interval $(2, 3]$, there is a constant $c$ independent of $\delta$ such that there is no $(1+c)$-approximation for \disp{} for $\delta \in (2, 3]$.
Thus with the translation reduction we obtain a constant $c'$, again independent of $\delta$, such that there is no $(1+c')$-approximation  for \disp{} for $\delta \in (\frac{2}{3}, \frac{3}{4}]$.
This explains the lack of an approximation scheme with an approximation factor approaching $1$ as $\delta$ approaches $\frac{2}{3}$ from above, as mentioned in \Cref{sec:approximations}.
On the other hand, the Subdivision Reduction that is used repeatedly to achieve the hardness result for $\delta \rightarrow \infty$ pushes the factor closer to $1$ each time, which is then translated to the values of $\delta$ as $\delta$ approaches $1$ from below.
Additionally, the existence of our approximation scheme suggests that this is asymptotically best possible.

Next, we show \APX-hardness for the missing intervals, namely $(\frac{1}{x}, \frac{2}{2x-1})$ for $x \in \mathbb N^+$.
To do so, we construct a chain of reductions to show hardness for the interval $(1, 2)$, and then use the Translation Reduction to extend that hardness result to all the missing intervals.
Notably, the reduction chain to approach $\delta = 1$ from above is infinitely long, which corresponds to the existence of our approximation schemes from \Cref{sec:approximations}.

\begin{lemma}
    \label{lemma:apxHardnessNumeratorOne}
    For $\delta \in (\frac{1}{x}, \frac{2}{2x-1})$ and $x \in \mathbb N^+$, \disp{} is \APX-hard.
\end{lemma}
\begin{proof}
    We first show hardness for every $\delta \in (1,2)$.
    Let $I_\ell \coloneqq (\frac{2(\ell+1)}{2\ell+1}, \frac{\ell+1}{\ell})$ for $\ell \in \mathbb N^+$.
    Then $I_\ell$ is the interval we obtain by applying the Translation Reduction $\ell$ times to the interval $(2, \infty)$ and then applying the Subdivision Reduction $\ell+1$ times.
    First, we observe the following:
    $$\frac{2((\ell+1)+1)}{2(\ell+1)+1} < \frac{2(\ell+1)}{2\ell+1} < \frac{(\ell+1)+1}{(\ell+1)} < \frac{\ell+1}{\ell}$$
    The first and the last inequality are obvious.
    For the second inequality, $\frac{2(\ell+1)}{2\ell+1} < \frac{(\ell+1)+1}{(\ell+1)} \Leftrightarrow 2(\ell+1)^2 < (2\ell+1)(\ell+1)+2\ell+1 \Leftrightarrow 0 < \ell$, which is true for all $\ell \in \mathbb N^+$.
    It thus follows that the intersection of $I_\ell$ and $I_{\ell+1}$ is non-empty.
    Further, for $\ell = 1$ we have $\frac{\ell+1}{\ell} = 2$ and additionally $\lim_{\ell \rightarrow \infty}\frac{2(\ell+1)}{2\ell+1} = 1$.
    Thus $\bigcup_{\ell \in \mathbb N^+} I_\ell = (1, 2)$, and it follows that \disp{} is \APX-hard for every $\delta \in (1,2)$.
    By applying the Translation Reduction repeatedly to the interval $(1,2)$, we obtain the claimed result.
\end{proof}

From \Cref{lemma:apxHardnessNumeratorTwo,lemma:apxHardnessNumeratorOne} as well as \Cref{sec:approximations} we can conclude the following theorem:

\begin{theorem}
    \label{thm:apxHardnessSmallerTwo}
    For every $\delta < 2$ that cannot be written as $\frac{1}{x}$ or $\frac{2}{x}$ for an $x \in \mathbb N^+$, \disp{} is \APX-complete.
\end{theorem}

\newpage
\bibliography{bibliography}

\appendix

\end{document}